\documentclass{article}
\usepackage{PRIMEarxiv}

\usepackage[utf8]{inputenc} 
\usepackage[T1]{fontenc}    
\usepackage{url}            
\usepackage{booktabs}       
\usepackage{amsfonts}       
\usepackage{nicefrac}       
\usepackage{microtype}      
\usepackage{lipsum}
\usepackage{fancyhdr}       
\usepackage{float}
\usepackage{graphicx}       
\graphicspath{{media/}}     
\usepackage{algpseudocode}
\usepackage[ruled,vlined,linesnumbered]{algorithm2e} 

\usepackage{amsfonts}
\usepackage{amsmath}
\usepackage{amssymb}
\usepackage{amsthm}
\usepackage{boxedminipage}
\usepackage{caption} 
\usepackage{enumitem}
\usepackage{graphicx}
\usepackage{multirow}
\usepackage{natbib}  
\usepackage{pifont}
\usepackage{pgfplots}
\usepackage{comment}
\usepackage{xcolor}         
\usepackage{tikz}
\usepackage{thm-restate}
\setlist{nolistsep}
\usetikzlibrary{arrows}
\usetikzlibrary{patterns}

\usetikzlibrary{shapes.geometric, arrows.meta, positioning}

\newtheorem{definition}{Definition}%
\newtheorem{theorem}{Theorem}%
\newtheorem{proposition}{Proposition}%
\newtheorem{example}{Example}

\pgfplotsset{compat=1.18}

\allowdisplaybreaks

\usepackage{authblk}

\title{\bf Compatible $k$-Relaxations of Fairness and Non-Wastefulness Under Hereditary Constraints
}

\author[1]{Tenma Wakasugi}
\author[1]{Zhaohong Sun}
\author[1]{Kei Kimura}
\author[1]{Makoto Yokoo}

\affil[1]{Kyushu University}

\date{\vspace{-10mm}}

\begin{document}

\maketitle

\begin{abstract}
We study two-sided matching markets under hereditary constraints, which extend beyond simple capacity limits and arise in applications such as diversity requirements and refugee resettlement. In these settings, fairness and non-wastefulness are often incompatible, and existing approaches typically address this tension by prioritizing one property at the expense of the other.
We take a different approach by relaxing both properties simultaneously in a controlled and symmetric manner. We introduce two notions indexed by an integer $k$: envy-received up to $k$ peers (ER-$k$) and non-wastefulness up to $k$ objections (NW-$k$). Our main theoretical result shows that ER-$k$ and NW-$k$ are always compatible under hereditary constraints for any fixed $k$. We provide two equivalent polynomial-time algorithms to compute such matchings: a $k$-admissible cutoff algorithm and a $k$-admissible college-proposing deferred acceptance mechanism. Finally, experimental results demonstrate that even small relaxations achieve a favorable balance between fairness and non-wastefulness.
\end{abstract}

\section{Introduction}

Centralized algorithms for two-sided matching markets have achieved remarkable success in a wide range of applications, most notably in matching residents to hospitals \citep{Roth84a} and students to colleges \citep{AbSo03b}. Classical models, however, typically assume simple capacity constraints on each institution. In many real-world matching and allocation problems, feasibility requirements are substantially richer and cannot be adequately represented by capacities alone.

A prominent example arises in college admissions, where diversity considerations play a central role. Students are often associated with observable characteristics, and colleges reserve seats by type in order to promote balanced representation \citep{AyTu17a,BCC+19b,Aybo21a}. In refugee resettlement, authorities must assign refugee families to localities subject to limited resources across multiple dimensions, such as education, healthcare, housing, and employment, while ensuring that all members of a family are placed together \citep{DKT23a,ACGS18a}. Daycare assignment provides another illustration: legal regulations impose age-dependent teacher–child ratios and space requirements, giving rise to feasibility constraints that cannot be effectively captured by simple capacity limits \citep{KaKo23a,STM+23a}. More complex feasibility structures also arise when individual institutions face their own capacity constraints while collections of institutions are jointly subject to aggregate quotas, as in Hungarian college admissions \citep{BFIM10a} and hospital–doctor matching in Japan \citep{KaKo15a}. Related feasibility concerns appear in internship and research matching programs, where assigning students to projects requires supervisors to allocate limited funding or supervision capacity across multiple projects, as observed in Australian undergraduate research schemes \citep{ABB24a}.

The abovementioned forms of feasibility constraints all fall within a broad framework commonly referred to as hereditary constraints \citep{GKK+15a,KTY18a}. This framework provides a unifying and expressive way to model a wide range of practically important matching environments. Under hereditary constraints, any feasible matching remains feasible after the removal of any subset of matched students. In this paper, we focus on matching problems with hereditary constraints and investigate which desirable properties can be achieved under such general feasibility restrictions.

Given the generality of these constraints, the classical notion of stability from the seminal matching literature is no longer guaranteed to exist. Stability can be decomposed into two distinct components: \emph{fairness}, which requires that no agent has justified envy toward another agent with lower priority, and \emph{non-wastefulness}, which requires that no agent can claim an empty seat, i.e., no agent can be matched to a more preferred option without violating feasibility. Under general feasibility constraints, these two properties are often fundamentally in tension.

Existing approaches typically address this tension by prioritizing one property at the expense of the other, either abandoning one property entirely \citep{GKK+15a,KaKo23a} or preserving one while substantially weakening the other \citep{ABB24a}. Consequently, mechanisms that enforce fairness may lead to highly wasteful outcomes with many empty seats, while mechanisms that enforce non-wastefulness may produce outcomes exhibiting widespread justified envy. 
What is notably missing is an approach that treats fairness and non-wastefulness in a balanced manner.

In this paper, we take a different approach. Rather than insisting on fully preserving one property, we propose to relax both fairness and non-wastefulness simultaneously, in a controlled and symmetric manner. By weakening each notion just enough, we restore their compatibility and obtain outcomes that strike a meaningful balance: matchings that are substantially fair while avoiding excessive waste. This perspective enables us to navigate the fundamental trade-off between fairness and efficiency under general feasibility constraints, leading to new existence results and polynomial-time algorithms.

To operationalize this idea, we introduce two symmetric, parametric relaxations of fairness and non-wastefulness. Both notions are indexed by a nonnegative integer $k$, which measures the extent to which violations are permitted.
The first notion, \emph{envy-received up to $k$ peers} (ER-$k$), relaxes fairness by limiting how widespread envy can be. Specifically, no student may be envied by more than $k$ other students. When $k=0$, ER-$0$ coincides with classical fairness, while larger values of $k$ allow controlled violations of fairness, in the sense that any envy is confined to at most $k$ affected students.
\footnote{This concept was studied independently and contemporaneously by \citet{TKK25a} in a different setting, and their paper appeared shortly before the present submission. Although motivated by different considerations, the resulting notion is equivalent.}
The second notion, \emph{non-wastefulness up to $k$ objections} (NW-$k$), relaxes non-wastefulness by limiting when a student can legitimately claim an empty seat. Under NW-$k$, an empty seat may remain unfilled only if assigning a student to that seat would trigger objections from more than $k$ other students. Thus, larger values of $k$ make it harder to justify leaving seats unfilled, strengthening efficiency requirements by allowing empty seats only when filling them would substantially harm other students.

We summarize our contributions as follows.

First, we introduce new relaxations of fairness and non-wastefulness, both parameterized by an integer $k$: \emph{Envy-Received-Freeness up to $k$ Agents} (ER-$k$) and \emph{Non-Wastefulness up to $k$ Objections} (NW-$k$). By varying $k$, these notions subsume several existing concepts in the literature and provide a unified taxonomy of fairness and non-wastefulness. Figure~\ref{fig:relationship} illustrates the relationships and compatibility among these notions.

Second, we establish a general existence result for ER-$k$ and NW-$k$ outcomes, showing that the two notions are always compatible under hereditary constraints. This result generalizes the findings of \citet{ABB24a}, whose notions of fairness and cutoff non-wastefulness correspond to the special cases ER-$0$ and NW-$0$ in our framework. We further show that this result is tight: in general, ER-$k$ and NW-$k+1$, as well as ER-$k-1$ and NW-$k$, are incompatible. Note that ER-k becomes more stringent as k decreases, whereas NW-k becomes more stringent as k increases.

Third, we design two polynomial-time algorithms for computing ER-$k$ and NW-$k$ matchings: a \emph{$k$-admissible cutoff algorithm} and a \emph{$k$-admissible college-proposing deferred acceptance mechanism}. The former extends the original cutoff algorithm, which is restricted to ER-$0$, to handle arbitrary values of $k$. We further show that the two algorithms are equivalent, with the latter providing an alternative interpretation grounded in the classical deferred acceptance framework.

Finally, we conduct an experimental evaluation of the $k$-admissible cutoff algorithm to examine the effect of varying $k$ on the resulting matchings. Our results show that even for small values of $k$, the algorithm achieves a favorable balance between fairness and non-wastefulness.

\begin{figure}[tb]
\centering
\begin{tikzpicture}[node distance=2cm, every node/.style={align=center}]

\node (a) at (2, 3) {stability};
\node (b) at (0, 1.5) {fairness = ER-$0$};
\node (c) at (4, 1.5) {non-wastefulness = NW-$|S|$};
\node (d) at (-2, 0) {EF-$k$};
\node (e) at (0, 0) {ER-$k$};
\node (f) at (0,-1.5) {ER-$|S|$};
\node (g) at (4,0) {NW-$k$};
\node (h) at (4,-1.5) {CNW = NW-$0$};

\draw[->, thick] (a) -- (b); 
\draw[->, thick] (a) -- (c); 
\draw[->, thick] (b) -- (d); 
\draw[->, thick] (b) -- (e); 
\draw[->, thick] (d) -- (f); 
\draw[->, thick] (e) -- (f); 
\draw[->, thick] (c) -- (g);
\draw[->, thick] (g) -- (h);

\draw[dashed, double, thick, red] (f) -- (c);
\draw[dashed, double, thick, red] (e) -- (g);
\draw[dashed, double, thick, red] (b) to (h); 

\node at (7, 0.3) {$\longrightarrow$ A implies B};
\node at (7, -0.2) {\textcolor{red}{==} compatibility};

\end{tikzpicture}
\caption{Relationship and compatibility of different fairness and non-wastefulness concepts with $0 \leq k \leq |S|$ where $|S|$ denotes the number of students. When $k \ge |S|$, the notions ER-$k$ (respectively, NW-$k$) are equivalent to ER-$|S|-1$ (respectively, NW-$|S|-1$).}
\label{fig:relationship}
\end{figure}
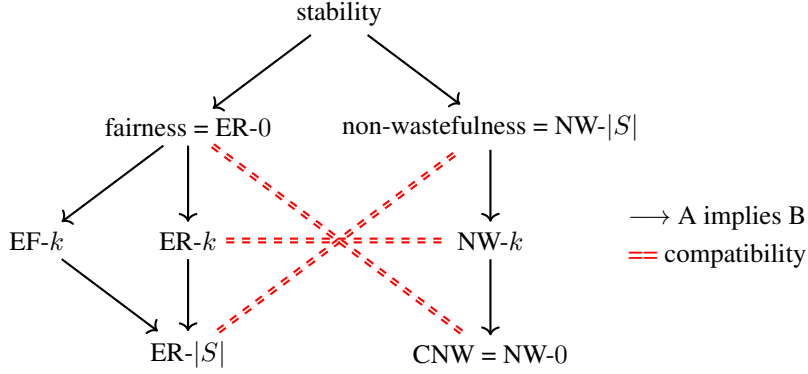

\section{Related Work}
In this section, we briefly review the literature most closely related to our paper.

A number of studies investigate matching under hereditary constraints. \citet{GKK+15a} propose the Adaptive Deferred Acceptance mechanism, which guarantees feasibility and non-wastefulness but does so at the expense of fairness. In contrast, \citet{KaKo23a} study a student-optimal fair matching algorithm that focuses exclusively on fairness and does not impose any form of non-wastefulness. Similarly, \citet{ABB24a} treat fairness as a fundamental requirement and introduce a weakened notion of cutoff non-wastefulness that remains compatible with fairness. In our framework, cutoff non-wastefulness corresponds to NW-0, which is the weakest notion in the NW-k hierarchy, as illustrated in Figure~\ref{fig:relationship}. Finally, \citet{CKM+22a} show that no strategyproof mechanism can satisfy fairness together with a weaker efficiency property, termed weak non-wastefulness, under hereditary constraints.

More recent work explores weakened notions of fairness by bounding the extent of justified envy. For example, \citet{CKL+24a} propose envy-freeness up to $k$ peers (EF-$k$), which limits the number of agents that a given agent may justifiably envy. From a complementary perspective, \citet{TKK25a} study a related notion that bounds the number of agents who may envy a given agent by $k$. However, these relaxations of fairness remain incompatible with meaningful notions of non-wastefulness, and the corresponding relaxations of non-wastefulness considered in these works are very weak.

A related line of work studies matching under distributional constraints using tools from discrete convex analysis. In particular, \citet{KTY18a} model feasibility constraints via $M^\natural$-concave functions and show that, under this structure, a generalized deferred acceptance mechanism produces stable outcomes and runs in polynomial time. This approach emphasizes mechanism design and incentive properties under specific convexity assumptions. In contrast, our work studies matching under general hereditary constraints and addresses the incompatibility between fairness and non-wastefulness through symmetric parametric relaxations, without relying on discrete convexity.

\section{Preamble}

\subsection{Model}

We formulate our model within the framework of college choice, which can be readily extended to other matching environments. An instance is represented by a tuple $(S, C, X, \succ_S, \succ_C, f)$, where $S$ is the set of students, $C$ is the set of colleges, $X \subseteq S \times C$ is the set of available contracts, $\succ_S$ and $\succ_C$ denote the preference relations of students and colleges, respectively, and $f$ specifies a feasibility function.

A contract $(s,c) \in X$ specifies the assignment of student $s$ to college $c$. 
We use $(s,\emptyset)$ to indicate that $s$ remains unassigned and $(\emptyset,c)$ to indicate that $c$ leaves one seat vacant. 
For any $Y \subseteq X$, let $Y_s := \{(s,c) \in Y \mid c \in C\}$ and $Y_c := \{(s,c) \in Y \mid s \in S\}$ denote, respectively, the set of contracts in $Y$ involving student $s$ and the set of contracts in $Y$ involving college $c$.

Each student $s \in S$ has a strict preference relation $\succ_s$ over contracts $X_s$. A contract $(s,c)$ is said to be \emph{acceptable} to student $s$ if $(s,c)\succ_s (s,\emptyset)$. 
Similarly, each college $c \in C$ has a strict priority order $\succ_c$ over contracts $X_c$. A contract $(s,c)$ is \emph{acceptable} to college $c$ if $(s,c)\succ_c (\emptyset,c)$. 
Unless otherwise specified, we assume that every contract $(s,c) \in X$ is acceptable 
to both student $s$ and college $c$, since any unacceptable contract can be removed from $X$ 
without loss of generality.
When convenient, we abbreviate $(s,c)\succ_s (s,c')$ as $c \succ_s c'$, 
and $(s,c)\succ_c (s',c)$ as $s \succ_c s'$.

A set of contracts $Y \subseteq X$ is a \emph{matching} if $|Y_s| \le 1$ for every $s \in S$, i.e., each student is involved in at most one contract.

\subsection{Feasibility Function}
In this subsection, we introduce a feasibility function $f$ that captures constraints extending beyond simple capacity limits. Under this framework, the feasibility of a matching is determined solely by the \emph{number} of students assigned to each college, rather than on their specific \emph{identities}.

Formally, let $m$ denote the number of colleges, and let $\mathbb{Z}_+^m$ be the set of $m$-dimensional vectors with nonnegative entries. A \emph{load vector} $\nu \in \mathbb{Z}_+^m$ specifies, for each college, the number of students assigned to it; that is, $\nu_i$ represents the number of students assigned to college $c_i$.

\begin{definition}[Feasibility Function and Induced Feasible Vectors]
\label{def:feasibility}  
A \emph{feasibility function} is a mapping $f: \mathbb{Z}_+^m \to \{0, -\infty\}$ that indicates whether a load vector is feasible.  
For a load vector $\nu \in \mathbb{Z}_+^m$, define  
\[
f(\nu) =
\begin{cases}
0, & \text{if $\nu$ satisfies all distributional constraints},\\[2mm]
-\infty, & \text{otherwise.}
\end{cases}
\]

The \emph{induced set of feasible vectors} $F$ is the set of load vectors that are feasible according to $f$, given by
\[
F = \{\nu \in \mathbb{Z}_+^m \mid f(\nu) = 0\}.
\]
\end{definition}

A fundamental type of distributional constraint is the maximum quota, where each college has a capacity on the number of students it can accept. We use maximum quotas to illustrate the feasibility function $f$.

\begin{example}[Feasibility Function under Maximum Quotas] 
Suppose each college $c \in C$ has a maximum quota $q_c$.  
Let $\nu$ denote a load vector, where $\nu_i$ represents the number of students assigned to college $c_i$.  
The feasibility function $f$ is then defined as
\begin{equation}
f(\nu) =
\begin{cases}
0, & \text{if each college $c$ admits no more than $q_{c}$ students},
\\
-\infty, & \text{otherwise}.
\end{cases}
\end{equation}

The induced set of feasible vectors is
\[
F = \{\nu \in \mathbb{Z}_+^m \mid \nu_i \le q_{c_i} \text{ for all } i=1,\dots,m\}.
\]
\end{example}

\begin{definition}[Feasible Matching]
\label{def:feasible_matching}  
For a matching $Y \subseteq X$, let $\nu(Y) := (|Y_{c_1}|, \dots, |Y_{c_m}|)$ denote the vector of student counts at each college.  

Given a feasibility function $f$, we say that $Y$ is \emph{feasible} if $f(\nu(Y)) = 0.$
\end{definition}

\subsection{Hereditary Constraints}
In this paper, we focus on a class of feasibility constraints called \emph{hereditary constraints}. 
Intuitively, under such constraints, if a matching $Y$ is feasible, then it remains feasible when weakly fewer students are assigned to each college. 
In other words, removing students from any subset of colleges cannot render a feasible matching infeasible.

To define hereditary constraints formally, we first introduce a way to compare load vectors.  
We use the coordinatewise order on $\mathbb{Z}_+^m$: for $\nu, \nu' \in \mathbb{Z}_+^m$, we write $\nu \le \nu'$ if $\nu_i \le \nu'_i$ for all $i$, and $\nu < \nu'$ if $\nu \le \nu'$ and $\nu_i < \nu'_i$ for some $i$.  

\begin{definition}[Hereditary Constraints]
A set of load vectors $F \subseteq \mathbb{Z}_+^m$ is \emph{hereditary} if 
\begin{enumerate}
    \item $\mathbf{0} \in F$ and,
    \item  for all $\nu, \nu' \in \mathbb{Z}_+^m$, $\nu \in F$ and $\nu > \nu'$ imply $\nu' \in F$.  
\end{enumerate}

A feasibility function $f$ is \emph{hereditary} if its induced set 
$F = \{\nu \in \mathbb{Z}_+^m : f(\nu) = 0\}$ is hereditary.
\end{definition}

\begin{example}[Hereditary Feasibility with Regional Quota]
\label{example:hereditary}
Consider two colleges $c_1$ and $c_2$ with a regional quota of $4$ students in total.  
Let $\nu = (\nu_1, \nu_2) \in \mathbb{Z}_+^2$ denote a load vector, where $\nu_i$ is the number of students assigned to college $c_i$.  

The feasibility function $f$ is defined as
\[
f(\nu) =
\begin{cases}
0, & \text{if } \nu_1 + \nu_2 \le 4,\\
-\infty, & \text{otherwise}.
\end{cases}
\]

The induced set of feasible vectors is
\[
F = \{\nu \in \mathbb{Z}_+^2 \mid \nu_1 + \nu_2 \le 4\}.
\]

This set is \emph{hereditary}, since if a load vector $\nu \in F$, then any vector $\nu' \le \nu$ (coordinatewise) is also in $F$.  
Figure~\ref{fig:hereditary} illustrates the feasible region: the blue-shaded area represents all feasible load vectors, and the red point $\nu = (2,2)$ is an example of a feasible assignment. Any vector below and to the left of it is also feasible, while vectors outside the shaded region are infeasible.
\end{example}
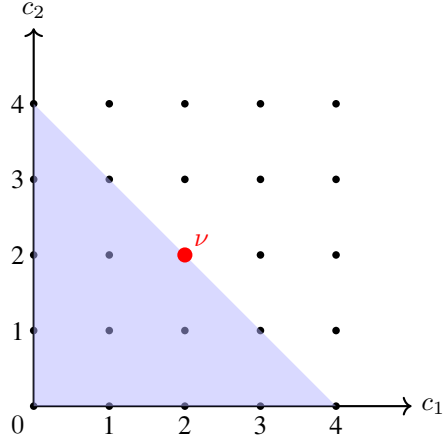
\begin{figure}[tb]
    \centering
\begin{tikzpicture}[scale=1]
\draw [->,thick] (0,0) -- (5,0) node [right] {$c_1$};
\draw [->,thick] (0,0) -- (0,5) node [above] {$c_2$};
\node [below left] at (0,0) {$0$};

\foreach \i in {0,...,4} {
    \foreach \j in {0,...,4} {
        \fill[black] (\i,\j) circle (0.05);
    }
}

\fill[blue!30,opacity=0.5]
(0,0) -- (0,4) -- (1,3) -- (2,2) -- (3,1) -- (4,0) -- cycle;

\fill[red] (2,2) circle (0.1) node[above right] {$\nu$};

\foreach \i in {1,...,4} \draw (\i,0) node[below] {\i};
\foreach \j in {1,...,4} \draw (0,\j) node[left] {\j};

\end{tikzpicture}
\caption{Illustration of hereditary constraints in Example~\ref{example:hereditary}. The shaded region represents feasible load vectors satisfying $\nu_1 + \nu_2 \le 4$. If $\nu$ (red point) is feasible, then any vector weakly below and to the left of $\nu$ is also feasible.}
\label{fig:hereditary}
\end{figure}

\subsection{Connection with Other Distributional Constraints}

In this subsection, we present several examples of constraints that can be expressed as hereditary constraints. These examples illustrate that the hereditary framework encompasses a range of existing and well-studied models in the matching and allocation literature.

\paragraph{Diversity in College Choice with Type-Specific Quotas.}
Many college districts impose diversity requirements on the composition of students admitted to each college. A standard modeling approach is to impose type-specific quotas, which limit the number of students of each type that a college may admit \citep{AbSo03b,EHYY14a,AzSu25a}. Such constraints are hereditary: if a set of admitted students satisfies all type-specific quotas, then any subset of that set also satisfies the quotas. Hence, diversity constraints in college choice naturally fall within the class of hereditary constraints.

\paragraph{Refugee Resettlement with Multidimensional Constraints.}
Refugee resettlement is one of the most pressing social challenges worldwide. In these problems, authorities must assign refugee families to localities, subject to the requirement that all members of a family be placed together. In addition, different families require varying types and amounts of services, such as job training, education, healthcare, and housing, while each locality has limited resources across multiple dimensions. Because family sizes and service needs differ, these feasibility constraints cannot, in general, be captured by a single capacity constraint. Such requirements are commonly modeled as multidimensional constraints and have been studied by \citet{DKT23a,ACGS18a}. Importantly, multidimensional constraints are hereditary: if an assignment is feasible with respect to all resource dimensions, then any subset of that assignment remains feasible.

\paragraph{Daycare Matching with Teacher–Child Ratio Constraints.}
In many countries, the assignment of daycare seats for young children is organized by local governments through centralized matching mechanisms. Daycare centers are subject to legal regulations governing teacher–child ratios. Because these ratios vary by the age of the children, the resulting feasibility constraints cannot be represented by simple capacity limits. For example, Japanese national regulations \citep{KaKo23a,STM+23a} require at least one teacher for every three children aged 0, one teacher for every six children aged 1 or 2, one teacher for every twenty children aged 3, and one teacher for every thirty children aged 4 or 5. These age-dependent staffing requirements are hereditary: if an assignment satisfies all teacher–child ratio constraints, then any subset of that assignment also satisfies them.

\paragraph{Matching with Common Quotas.}
Matching with common quotas has been studied in the context of Hungarian college admissions, where individual colleges have upper quotas and certain bounded subsets of colleges are subject to common quotas that are smaller than the sum of their individual quotas \citep{BFIM10a}. Similarly, in the Japanese residency matching program, the government introduced regional caps that restrict, for each of the 47 prefectures, the total number of residents matched within the prefecture. The common quotas model is hereditary, as feasibility is preserved under the removal of any subset of matched agents.

\paragraph{Project Matching with Budget Constraints.}
This model captures the assignment of student applicants to projects proposed by supervisors in internship or research programs. A distinctive feature is that assigning a student to a project requires a contribution from the supervisor’s limited budget or time. One motivating application is summer undergraduate research programs in Australia, where students undertake research projects supervised by one or more faculty members \citep{ABB24a}. Individual supervisors may offer multiple projects; however, even when projects are partially subsidized by institutional funding, supervisors are typically required to contribute from their own research budgets or allocate limited supervisory time. These supervisor-side constraints imply that not all projects can be supported simultaneously, and feasibility therefore cannot be characterized by simple capacity constraints. Related models of matching with budget constraints are studied by \citet{KaIw17a,KaIw18a,IHZ+19a}.

\subsection{Fundamental Properties}

In this subsection, we recall several fundamental properties from matching theory.  
One of the central concepts is \emph{stability}, introduced by \citet{GaSh62a}.  
A matching is said to be stable if it admits no \emph{blocking pair}, that is, no pair of agents who would both prefer to be matched with each other rather than remain with their current assignments.  

Blocking pairs can be classified into two types: those arising from \emph{justified envy}, where a student could displace another student with lower priority at a college, and those arising from a student’s ability to \emph{claim an empty seat} by occupying a vacant position.  
The elimination of the former corresponds to fairness, while the elimination of the latter corresponds to non-wastefulness.  
Consistent with the classical matching literature, the conjunction of fairness and non-wastefulness is equivalent to stability.  
We now adapt these notions to our framework with hereditary constraints.

\begin{definition}[Justified Envy and Fairness] 
\label{def:envy}
In a matching $Y$, a student $s$ has \emph{justified envy} toward another student $s'$ if there exists a college $c$ such that:
\begin{enumerate}
    \item $(s, c) \succ_s Y_s$, i.e., student $s$ strictly prefers college $c$ over their current assignment $Y_s$;
    \item $(s',c) \in Y$ and $s \succ_c s'$, i.e., college $c$ admits student $s'$ even though it ranks $s$ higher than $s'$.
\end{enumerate}
A matching $Y$ is \emph{fair} if no student has justified envy.
\end{definition}

Equivalently, a matching is fair if, whenever a student is not assigned to a college they prefer, all students assigned to that college are ranked strictly above them in the college’s preference ordering.

\begin{definition}[Claiming an Empty Seat and Non-wastefulness (NW)]
\label{def:empty_seat}
Given a matching $Y$, a student $s$ is said to \emph{claim an empty seat} at a college $c$ if  
\begin{enumerate}
    \item $(s, c) \succ_s Y_s$, i.e., $s$ strictly prefers $c$ over their current assignment $Y_s$; and  
    \item the matching $Y' = (Y \setminus Y_s) \cup \{(s,c)\}$ is feasible.
\end{enumerate}
A matching $Y$ is \emph{non-wasteful} if no student claims an empty seat.
\end{definition}

Intuitively, a student $s$ claims an empty seat at a college $c$ if they prefer $c$ over their current assignment and could be reassigned to $c$ without violating feasibility.  
Non-wastefulness ensures that there are no ``wasted opportunities'': for every student $s$ and every college $c$ that $s$ prefers over $Y_s$, reassigning $s$ to $c$ would necessarily violate the feasibility constraints.

\begin{proposition}
\label{prop:nonexistence}
\citep{KaKo17b}
Under hereditary constraints, a matching that is both fair and non-wasteful need not exist.
\end{proposition}

\begin{proof}
We prove Proposition~\ref{prop:nonexistence} via the following counterexample.

\begin{example}
\label{example:nonexistence}
Consider two students $S = \{s_1, s_2\}$ and two colleges $C = \{c_1, c_2\}$.  
The preferences and priorities are:
\[
c_1 \succ_{s_1} c_2, 
\quad c_2 \succ_{s_2} c_1,
\quad 
s_2 \succ_{c_1} s_1, \quad 
s_1 \succ_{c_2} s_2
\]

Suppose the set of feasible allocations is
\[
F = \{(\nu_1,\nu_2) \in \mathbb{Z}_+^2 \mid \nu_1 + \nu_2 \le 1 \},
\] 
which imposes a \emph{regional cap}: at most one student can be matched in total.  
The four possible non-empty matchings are:
\begin{itemize}
    \item $Y_1 = \{(s_1, c_1)\}$: not fair, since $s_2$ envies $s_1$ at $c_1$;
    \item $Y_2 = \{(s_1, c_2)\}$: $s_1$ can claim an empty seat at $c_1$;
    \item $Y_3 = \{(s_2, c_1)\}$: $s_2$ can claim an empty seat at $c_2$;
    \item $Y_4 = \{(s_2, c_2)\}$: not fair, since $s_1$ envies $s_2$ at $c_2$.
\end{itemize}

Thus, no matching is simultaneously fair and non-wasteful.
\end{example}

This completes the proof of Proposition~\ref{prop:nonexistence}.
\end{proof}



\section{Weakening Fairness and Non-wastefulness Simultaneously}
Since fairness and non-wastefulness may be incompatible under hereditary constraints, a natural approach is to consider relaxed variants of both properties.
In this section, we introduce two notions that simultaneously weaken fairness and non-wastefulness, each parameterized by an integer $k$, and show that these relaxed concepts are always mutually compatible.  
Our framework strictly generalizes existing results; in particular, it encompasses the existence theorem of \citet{ABB24a} as the special case $k = 0$ for both fairness and non-wastefulness.

\subsection{Envy-Received up to $k$ Peers (ER-$k$)}
Instead of requiring the complete absence of justified envy, recent research has proposed relaxed notions of fairness that bound the extent of justified envy. 
One such notion is \emph{envy-freeness up to $k$ peers (EF-$k$)}, which requires that each student envies at most $k$ other students \citep{CKL+24a}.

\begin{definition}[Envy-Freeness up to $k$ Peers (EF-$k$)] 
\label{EF-k}  
Given a feasible matching $Y$ and a student $s$, define  
\[
Ev(Y,s) = \{\, s' \in S \mid \exists c \in C : (s',c) \in Y,\; c \succ_{s} Y_{s},\; s \succ_c s' \,\},  
\]
that is, the set of students whom $s$ envies under $Y$.  
A feasible matching $Y$ satisfies \emph{EF-$k$} if, for every student $s \in S$, we have $|Ev(Y,s)| \leq k$.
\end{definition}

In contrast, we adopt a different relaxation that reverses the perspective. Rather than limiting the number of students a given student envies, we bound the amount of envy that each student receives. Specifically, a matching satisfies \emph{envy-received up to k peers} (ER-$k$) if no student is envied by more than k other students.

\begin{definition}[Envy-Received up to $k$ Peers (ER-$k$)]  
\label{ER-$k$}  
Given a feasible matching $Y$ and a student $s$, define  
\[
Evd(Y,s) = \{\, s' \in S \mid \exists c \in C : (s,c) \in Y,\; c \succ_{s'} Y_{s'},\; s' \succ_c s \,\},  
\]  
that is, the set of students who envy $s$ under $Y$.  
A feasible matching $Y$ satisfies \emph{ER-$k$} if, for every student $s \in S$, we have $|Evd(Y,s)| \leq k$.
\end{definition}

From another perspective, ER-$k$ ensures that for each college $c$, the total number of students who have justified envy toward some student assigned to $c$ is at most $k$, thereby bounding the intensity of envy within each college. Moreover, any feasible matching trivially satisfies ER-$(n-1)$, where $n$ denotes the number of students.

By definition, if a matching satisfies ER-$k$, then it also satisfies ER-$\ell$ for any $\ell > k$: if no agent is envied by more than $k$ agents, then in particular no agent is envied by more than $\ell$ agents.

\begin{proposition}  
\label{prop:monotone:ERk}  
For any integers $\ell > k$, every feasible matching $Y$ that satisfies ER-$k$ also satisfies ER-$\ell$.
\end{proposition}

Both EF-$k$ and ER-$k$ reduce to fairness when $k=0$, but they represent different relaxations for larger $k$.  

\begin{proposition}  
\label{prop:connection:EFk&ERk}  
For $k>0$, EF-$k$ does not imply ER-$k$, and ER-$k$ does not imply EF-$k$. 
\end{proposition}  

\begin{proof}  
Consider a market with three students $S = \{s_1, s_2, s_3\}$ and one college $C = \{c_1\}$ with two seats.  
The college’s preference ordering is $\succ_{c_1}: s_1 \succ s_2 \succ s_3$.  

Take $k = 1$. First, consider the matching 
\[
Y_1 = \{(s_2, c_1), (s_3, c_1)\}.
\]  
Here, both $s_2$ and $s_3$ are envied by the unmatched student $s_1$.  
Since each student is envied by at most one other student, $Y_1$ satisfies ER-$1$.  
However, $s_1$ envies both $s_2$ and $s_3$, so EF-$1$ is violated.  
Hence, ER-$k$ does not imply EF-$k$.  

Conversely, consider the matching 
\[
Y_2 = \{(s_3, c_1)\}.
\]  
Now, both $s_1$ and $s_2$ envy $s_3$, so $s_3$ is envied by two students, violating ER-$1$.  
However, each of $s_1$ and $s_2$ envies only one student ($s_3$), so EF-$1$ is satisfied.  
Hence, EF-$k$ does not imply ER-$k$.  

This completes the proof of Proposition~\ref{prop:connection:EFk&ERk}.  
\end{proof}

\subsection{Non-Wastefulness up to $k$ Objections 
(NW-$k$)}

We next introduce a relaxed version of the non-wastefulness concept. Recall that a feasible matching is non-wasteful if no student can claim an empty seat at a college they prefer. Weakening non-wastefulness inevitably permits some claims on empty seats to be considered legitimate. To ensure compatibility with ER-$k$, we allow such a claim to be dismissed if fulfilling it would violate the ER-$k$ condition.

To this end, we first define the notion of an \emph{objection} to a student’s claim on an empty seat, which forms the basis of our new concept.

\begin{definition}[Objection to a Claim on an Empty Seat]
\label{def:objection}
Consider a feasible matching $Y$ in which a student $s$ claims an empty seat at college $c$.  
A student $s'$ has an \emph{objection} to $s$’s claim on the empty seat at $c$ if the following conditions hold:  
\begin{enumerate}
    \item $s' \succ_c s$, i.e., $s'$ has higher priority than $s$ at college $c$;  
    \item $c \succ_{s'} Y_{s'}$, i.e., $s'$ prefers $c$ over their current assignment $Y_{s'}$;  
    \item $(Y \setminus Y_{s'}) \cup \{(s',c)\}$ is infeasible, i.e., $s'$ cannot be reassigned to $c$ without violating feasibility.  
\end{enumerate}
\end{definition}

Intuitively, if student $s$ claims an empty seat at college $c$, an objection from student $s'$ indicates that assigning $s$ to $c$ would cause $s'$ to have justified envy toward $s$, as illustrated in the following example.

\begin{example}[Illustration of Objection]
\label{example:objection}
Revisit Example~\ref{example:nonexistence} in the proof of Proposition~\ref{prop:nonexistence}, where two students compete for two colleges under a regional quota of 1, so that only one student can be admitted.

Consider the matching $Y_2 = \{(s_1, c_2)\}$, in which student $s_1$ claims an empty seat at the more preferred college $c_1$. However, $s_2$ has higher priority than $s_1$ at $c_1$ ($s_2 \succ_{c_1} s_1$), and assigning $s_2$ to $c_1$ would violate feasibility due to the regional cap of 1.  
Thus, if student $s_1$ were assigned to $c_1$ instead of $c_2$, student $s_2$ would have justified envy toward $s_1$.  For this reason, we say that the unmatched student $s_2$ \emph{objects} to $s_1$’s claim of the empty seat at $c_1$.
\end{example}

Building on the notion of objection, we now introduce a relaxed form of non-wastefulness that specifies the minimum number of potential objections to an empty-seat claim, parameterized by an integer $k$. Intuitively, it allows an empty seat to remain unfilled as long as admitting a student into that seat would trigger sufficiently many objections.

\begin{definition}[Non-wastefulness up to $k$ Objections (NW-$k$)]
\label{def:NW_k}
Given a feasible matching $Y$, a claim by a student $s$ to an empty seat at a college $c$ is called \emph{$k$-legitimate} if at most $k$ students object to this claim.

A feasible matching $Y$ satisfies \emph{non-wastefulness up to $k$ objections} (NW-$k$) if there exists no $k$-legitimate claim.
\end{definition}

In words, a feasible matching $Y$ is NW-$k$ if, whenever a student wishes to occupy an empty seat at a college, assigning the student to that seat would trigger more than $k$ objections. Conversely, if filling an empty seat would generate at most $k$ objections, then leaving the seat unfilled is considered wasteful, and such a matching violates NW-$k$.

The parameter $k \in \mathbb{N}$ controls the tolerance for objections. Given a matching $Y$, a student’s claim to an empty seat may face anywhere from $0$ to $|S|-1$ objections. A \emph{$k$-legitimate} claim therefore measures how difficult it is to block the claim: larger values of $k$ require more objections to justify leaving a seat empty.

As $k$ increases, the NW-$k$ requirement becomes stronger, because a student’s claim to an empty seat can be blocked only if it triggers more than $k$ objections. A larger value of $k$ therefore allows a student to successfully claim an empty seat even in the presence of more objections, expanding the set of colleges the student can legitimately claim. When $k$ $=$ $0$, NW-$0$ merely requires that every unfilled seat trigger at least one objection, corresponding to a very weak notion of non-wastefulness. In other words, as $k$ grows, more claims become legitimate, increasing students’ opportunities to access preferred colleges. The following proposition formalizes this monotonic relationship.

\begin{proposition}[Monotonicity]
\label{prop:NWl}
For any integers $\ell < k$, if a feasible matching $Y$ satisfies NW-$k$, then $Y$ also satisfies NW-$\ell$.
\end{proposition}

\begin{proof}
Recall that a claim is called \emph{$k$-legitimate} if it has at most $k$ objections, and $Y$ satisfies NW-$k$ precisely when there is no $k$-legitimate claim in $Y$.

Let $\ell<k$ and suppose $Y$ satisfies NW-$k$. Assume, for contradiction, that $Y$ does not satisfy NW-$\ell$. Then there exists a claim in $Y$ that is $\ell$-legitimate, i.e. a claim with at most $\ell$ objections. But since $\ell<k$, any claim with at most $\ell$ objections also has at most $k$ objections, hence it is $k$-legitimate. This contradicts the assumption that $Y$ has no $k$-legitimate claim. Therefore no $\ell$-legitimate claim exists in $Y$, and $Y$ satisfies NW-$\ell$.
\end{proof}

The monotonicity result shows that the NW-$k$ conditions form a nested hierarchy, with a stronger condition implying all weaker ones:
\[
\text{NW-}n \;\Rightarrow\; \cdots \;\Rightarrow\; \text{NW-}k \;\Rightarrow\; \text{NW-}(k-1) \;\Rightarrow\; \cdots \;\Rightarrow\; \text{NW-}1 \;\Rightarrow\; \text{NW-}0
\]

In particular, when $k = n$, where $n$ denotes the total number of students, NW-$n$ coincides with the standard notion of non-wastefulness (NW), as shown in the following proposition.

\begin{proposition}[NW-$|S|$]
\label{prop:NWn&NW}
A feasible matching $Y$ satisfies NW-$|S|$ if and only if $Y$ satisfies NW.
\end{proposition}

\begin{proof}
Recall that a claim on an empty seat is called  $k$-legitimate if at most $k$ students have an objection to it, and a feasible matching satisfies NW-$k$ if no claim in $Y$ is $k$-legitimate.

We first show that NW-$|S|$ implies NW. 
Suppose $Y$ satisfies NW-$|S|$. Consider  any student $s$ who could claim an empty seat at some college $c$. By definition, the claim would be $|S|$-legitimate if at most $|S|$ students object. However, for any claim by $s$, the maximum number of students who could object is at most $|S|-1$, because $s$ cannot object to their own claim. Therefore, \emph{any potential claim would automatically be $|S|$-legitimate}. Since $Y$ satisfies NW-$|S|$, no $|S|$-legitimate claim exists in $Y$. Consequently, no student can claim an empty seat, and $Y$ satisfies the original NW.

We next show that NW implies NW-$|S|$. Conversely, suppose $Y$ satisfies NW. By definition, no student claims an empty seat. Hence, trivially, no claim in $Y$ is $|S|$-legitimate, so $Y$ also satisfies NW-$|S|$.

Thus, a feasible matching satisfies NW-$|S|$ if and only if it satisfies NW.
\end{proof}

At the other extreme, NW-$0$ coincides with the notion of \emph{cut-off non-wastefulness} (CNW) introduced by \citet{ABB24a}.
The three conditions of objection in Definition~\ref{def:objection} capture exactly the situation in which a higher-priority student would have justified envy if a claim on an empty seat were granted. These conditions were implicitly used in \citet{ABB24a}; here, we make them explicit as \emph{objections} and extend the idea from a single objecting student to a threshold of $k$ objections.
Using our terminology, CNW can be restated as follows:

\begin{definition}[Cut-off Non-Wastefulness (CNW) \citep{ABB24a}]
\label{def:cutoff_non_waste}
A matching $Y$ is \emph{cut-off non-wasteful} if either $Y$ is non-wasteful, or whenever a student $s$ claims an empty seat at a college $c$, there exists a student $s'$ who has an objection to the claim. 
\end{definition}

In other words, cut-off non-wastefulness allows a student $s$ to claim an empty seat at a college $c$ only if doing so does not generate justified envy for any higher-ranked student $s'$ whose reassignment to $c$ would be infeasible.


\begin{proposition}[NW-$0$]
\label{prop:NW0&CNW}
A feasible matching $Y$ satisfies NW-$0$ if and only if $Y$ satisfies CNW.
\end{proposition}

\begin{proof}
Under NW-$0$, a claim on an empty seat is $0$-legitimate if and only if it faces no objections. Therefore, $Y$ satisfies NW-$0$ if and only if there exists no claim on an empty seat that admits zero objections.

By Definition~\ref{def:cutoff_non_waste}, $Y$ satisfies CNW if and only if every claim on an empty seat is objected to by at least one student. Thus, the absence of a claim with zero objections is exactly equivalent to cutoff non-wastefulness.

Hence, $Y$ satisfies NW-$0$ if and only if $Y$ satisfies CNW.
\end{proof}

\subsection{Compatibility between ER-$k$ and NW-$k$}

We next present our first main theoretical contribution, which states the compatibility between \emph{ER-$k$} and \emph{NW-$k$}. This theorem demonstrates that fairness and non-wastefulness, when relaxed in a parallel and balanced way through the same parameter $k$, remain jointly attainable. It shows that the two properties are not in conflict but can be satisfied simultaneously across the entire range of relaxations.

\begin{theorem}
\label{thm:existence:ER&NWk}
Under hereditary constraints, for every integer $k$ with $0 \leq k \leq n$, there always exists a matching that is both ER-$k$ and NW-$k$.
\end{theorem}
\begin{proof}
We construct such a matching via the following iterative algorithm, described in Algorithm~\ref{alg:existence_ER_NWk}.

\begin{algorithm}[tb]
\caption{Construction of an ER-$k$ and NW-$k$ Matching}
\label{alg:existence_ER_NWk}
\KwIn{Students $S$, colleges $C$, feasibility function $f$, integer $k$}
\KwOut{A matching $Y$ that is both ER-$k$ and NW-$k$}

Initialize $Y \gets \emptyset$\;

\While{there exists a student--college pair $(s,c)$ satisfying the conditions below}{
    \begin{enumerate}
        \item $s$ can claim an empty seat at $c$;
        \item the updated matching $Y' = (Y \setminus Y_s) \cup \{(s,c)\}$ is ER-$k$.
    \end{enumerate}
    Update $Y \gets Y'$.
}

\Return $Y$\;
\end{algorithm}

\noindent
\textbf{Correctness.}  
We start from the empty matching, which is trivially ER-$k$ (though not necessarily NW-$k$).  
At each iteration, we add a pair $(s,c)$ only if the resulting matching remains ER-$k$.  
When no such admissible pair exists, the algorithm terminates.  
By construction, the final matching $Y$ satisfies:
\begin{enumerate}
    \item \textbf{ER-$k$:} every update preserves ER-$k$ by condition (2);
    \item \textbf{NW-$k$:} if any unassigned student could claim an empty seat without violating ER-$k$, the loop would continue, contradicting termination.
\end{enumerate}

\noindent
\textbf{Termination.}  
Throughout Algorithm~\ref{alg:existence_ER_NWk}, no student is ever rejected by a college, but some may move to a more preferred one.  
Since each student’s preference list is finite and every iteration strictly improves one student’s assignment, the algorithm must terminate after a finite number of steps.  
The resulting matching is therefore both ER-$k$ and NW-$k$.
\end{proof}

Our result strictly generalizes the existence theorem of \citet{ABB24a}, who show that fairness and cut-off non-wastefulness are compatible.
Their theorem corresponds to the special case $k = 0$ in our framework.

\begin{proposition}[\citet{ABB24a}]
\label{prop:fair&CNW}
Under hereditary constraints, there exists a matching that is both fair and cut-off non-wasteful, i.e., ER-0 and NW-0. 
\end{proposition}

We next establish the tightness of Theorem~\ref{thm:existence:ER&NWk}. Specifically, we show that under hereditary constraints, no allocation can simultaneously satisfy ER-$k$ and NW-$(k+1)$, nor ER-$(k-1)$ and NW-$k$. Note that ER-k becomes more stringent as k decreases, whereas NW-k becomes more stringent as k increases.

\begin{theorem}
\label{thm:incompatibility-ER-$k$-NW-k+1}
Under hereditary constraints, for any $0 \leq k \leq |S|-2$, ER-$k$ and NW-$(k+1)$ are incompatible. Likewise, ER-$(k-1)$ and NW-$k$ are incompatible.
\end{theorem}

\begin{proof}
Consider a scenario with $|S|$ students and $|S|$ colleges. 
For each student $s_i$, her preference list is:
\[
c_{i} \succ_{s_i} c_{i+1} \succ_{s_i} \ldots 
\succ_{s_i} c_{|S|} \succ_{s_i} c_1 
\succ_{s_i} \ldots \succ_{s_i} c_{i-1},
\]
where the student prefers colleges in a cyclic manner, with college $c_{i}$ being the most preferred, and $c_{i-1}$ being the least preferred.

For each college $c_i$, its preference list is:
\[
s_{i+1} \succ_{c_i} s_{i+2} \succ_{c_i} \ldots
\succ_{c_i} s_{|S|} \succ_{c_i} s_1 
\succ_{c_i} \ldots \succ_{c_i} s_{i},
\]
with the college's preferences arranged cyclically, where $s_{i+1}$ is the most preferred student and $s_{i}$ is the least preferred.

In summary, for each student $s_i$, her most preferred college $c_{i}$ considers her as the least preferred, while her least preferred college $c_{i-1}$ considers her as the most preferred.

The distributional constraint $f$ is defined as:
\[
f(\nu) = 0 \quad \text{iff} \quad \sum_{c_i \in C} |\nu_i| \leq 1,
\]
which means that the total number of students assigned to all colleges cannot exceed one. Clearly, $f$ induces a regional constraint, which is hereditary.

Now, a feasible matching $Y$ satisfies ER-$k$ if and only if 
(i) $Y$ is an empty matching, or (ii) there exists a student $s$ and a college $c$ such that $c$ is ranked $k+1$-th from the bottom (equivalently, $|S|-k$-th from the top) or worse in $s$'s preference list, and $Y = \{(s, c)\}$.
Indeed, suppose otherwise that a student, say $s_1$, is matched in $Y$ to college $c_{|S|-k-1}$, which is ranked $|S|-k-1$-th in $s_1$'s preference list.
Then, $s_1$ receives justified envies from students $s_{{|S|-k,|S|-k+1,\dots, s_{|S|}}}$, resulting in  $|Evd(Y,s_1)|=k+1$. Consequently, the matching $Y$ does not satisfy ER-$k$.

On the other hand, 
a feasible matching $Y$ satisfies NW-$\ell$ if and only if there exists a student $s$ and a college $c$ such that $s$ is ranked at most $\ell+1$-th in $c$'s preference list (equivalently, $c$ is ranked at most $|S| - \ell$-th in $s$'s preference list), and $Y = \{(s, c)\}$.
Indeed, suppose that a student, say $s_1$, is matched in $Y$ to college $c_{|S|-\ell+1}$, which is ranked $|S|-\ell+1$-th in $s_1$'s preference list.
In this case, $s_1$ claims an empty seat at college $c_{|S|-\ell}$, and this claim is $\ell$-legitimate, since moving $s_1$ to $c_{|S|-\ell}$ induces only $\ell$ justified envies towards $s_1$.
Hence, the matching $Y$ does not satisfy NW-$\ell$.

From the above, for a matching $Y$ to satisfy both ER-$k$ and NW-$\ell$, there must exist a student who is matched with a college that is ranked at least $|S| - k$-th (i.e., no better than the $|S| - k$-th) in their preference list, and at most $|S| - \ell$-th (i.e., among the top $|S| - \ell$) in their preference list. Therefore, it must hold that $|S| - k \le |S| - \ell$, which implies $\ell \le k$. Hence, ER-$k$ and NW-$(k+1)$ are incompatible.
\end{proof}

\section{$k$-Admissible Cutoff Algorithm}
In this section, we describe an efficient procedure for computing a matching that satisfies both ER-$k$ and NW-$k$.
As established in Theorem~\ref{thm:existence:ER&NWk}, such a matching can be obtained by iteratively extending an empty matching while maintaining ER-$k$ until the NW-$k$ condition is met. However, the theorem does not specify how these updates should be carried out, and multiple algorithmic implementations are possible.

We next introduce the \emph{$k$-admissible cutoff algorithm}, a concrete and computationally efficient realization of this process.
As the name suggests, it generalizes the original cutoff algorithm of \citet{ABB24a} by incorporating \emph{$k$-admissibility} constraints; when $k = 0$, it coincides exactly with their procedure.
In the following section, we present an alternative formulation based on the classic college-proposing deferred acceptance algorithm, offering a complementary perspective on the same underlying principles.

\subsection{Original Cutoff Algorithm}  
For completeness and comparison, we first present the original cutoff algorithm of \citet{ABB24a}.

Each college $c$ assigns a score to every student $s$ according to its priority order: specifically, if $s$ is ranked $k$-th by $c$, then $s$ receives a score of $|S| - k + 1$.
Given a cutoff score vector $B \in \{1, 2, \dots, |S|+1 \}^C$, a student $s$ is said to be \emph{admissible} to college $c$ if her score at $c$ is at least $B_c$.
For any cutoff vector $B$, let $Y^B$ denote the \emph{induced matching} obtained by assigning each student to her most-preferred admissible college (if any).

The algorithm proceeds as follows. Initially, all cutoff scores are set to $|S|+1$, and every student is unmatched.
Then, iteratively, the cutoff score of some college $c$ is decreased by one whenever the resulting induced matching remains feasible.
The process terminates when the cutoff vector $B$ is \emph{minimal}, meaning that decreasing any cutoff score further would make the induced matching infeasible.

\subsection{$k$-Admissibility and New Cutoff Notion}

Our new algorithm resembles the original cutoff algorithm but departs from it in a crucial way: students ranked below the cutoff may still be admitted, provided that their admission does not violate the ER-$k$ condition.

The key innovation lies in refining the notion of admissibility when evaluating potential reassignments. To formalize this refinement, we introduce two new concepts, beginning with the notion of \emph{$k$-admissibility}, which generalizes the previous admissibility concept used in the original cutoff algorithm.

\begin{definition}[$k$-Admissibility]
\label{def:kadmissibility}
Given a feasible matching $Y$ and an integer $k$, a student $s$ is said to be \emph{$k$-admissible} to a college $c$ (with respect to $Y$) if there exists a reassignment $Y' = (Y \setminus Y_s) \cup \{(s, c)\}$ such that:
\begin{enumerate}
    \item matching $Y'$ is feasible; and
    \item at most $k$ students envy $s$ under $Y'$, that is, there are at most $k$ students who both (a) prefer $c$ to their current assignment in $Y'$, and (b) have higher priority at $c$ than $s$.
\end{enumerate}
\end{definition}

In words, a student $s$ is $k$-admissible to a college $c$ if assigning $s$ to $c$ would preserve feasibility and generate no more than $k$ instances of justified envy. In other words, a student $s$ is $k$-admissible to a college $c$ if $s$ has a $k$-legitimate claim for an empty seat of $c$. When $k = 0$, this definition reduces to the standard notion of admissibility used in the original cutoff algorithm.

We next introduce the concept of a \emph{cutoff student}, which corresponds to the \emph{cutoff score} at each college. This notion generalizes, and importantly differs from, the one used in the original cutoff algorithm. Intuitively, the cutoff student at a college $c$ is the lowest-priority student $s$ such that i) every student ranked above or at $s$ in $c$'s priority order is assigned to a college they weakly prefer to $c$, and ii) the next student immediately below $s$, if any, is assigned to a college strictly worse than $c$. Notably, the cutoff student $s$ need not be matched to $c$ itself, but only to a college that is weakly better than $c$. 

As in the original cutoff algorithm, each college $c$ assigns a score to every student $s$ according to its priority ordering: if $s$ is ranked $k$-th by $c$, then $s$ receives a score of $|S| - k + 1$. For the score $|S| + 1$, we imagine a dummy student who has the highest possible priority.
  
\begin{definition}[Cutoff Student and Cutoff Score]
\label{def:cutoff_student}
Given a feasible matching $Y$ and a college $c$, a student $s$ is the \textbf{cutoff student} for $c$ with respect to $Y$ if the following conditions hold:
\begin{enumerate}
  \item 
  every student with weakly higher priority than $s$ at $c$ is assigned to a college they weakly prefer to $c$, i.e.,
  \[
    \forall t \in S, \; t \succeq_c s \Rightarrow Y_t \succeq_t c.
  \]
  \item 
  Let $s'$ denote the student ranked immediately below $s$ in $c$’s priority order (if such a student exists), i.e., $s \succ_c s'$ and there is no $t$ such that $s \succ_c t \succ_c s'$.  
  Then $s'$ is assigned to a college strictly worse than $c$, i.e.,
  \[
    c \succ_{s'} Y_{s'}.
  \] 
\end{enumerate}

Given the cutoff student $s$ at college $c$, the \textbf{cutoff score} of $c$, denoted by $B_c$, is defined as
\[
B_c = |S| - \text{rank}_c(s) + 1,
\] 
where $\text{rank}_c(s)$ denotes the position of $s$ in $c$’s priority ordering, with the highest-priority student assigned rank $1$.
\end{definition}

Unlike in the original cutoff algorithm, where each cutoff serves as a \emph{strict admission threshold} that prevents any student ranked below it from being admitted, our cutoff score is \emph{non-binding}. A student ranked below $B_c$ may still be assigned to $c$ provided that doing so preserves ER-$k$. In particular, the student ranked immediately below the cutoff student (if any) is currently assigned to an option strictly worse than $c$ (including $\emptyset$, representing being unmatched).

The cutoff $B_c$ therefore marks the boundary between students who are effectively satisfied relative to $c$ and those who might benefit from moving to $c$. More precisely, any student ranked at or above $B_c$ does not claim an empty seat at $c$, whereas a student can claim an empty seat at $c$ only if they lie below the cutoff score $B_c$. As we show shortly, $B_c$ serves as a \emph{reference point} that guides the search for $k$-admissible students, rather than as a hard constraint on admission.

\begin{example}[Illustration of Cutoff Student and Cutoff Score]
\label{example:cutoff_student}
Consider an instance with students $S = \{s_1, s_2, s_3\}$ and colleges $C = \{c_1, c_2\}$. Each college has capacity one, and the total number of matched students is at most two with no further feasibility constraints.
The students’ preferences and colleges’ priority orders are:
\begin{align*}
    s_1: &\; c_2 \succ c_1,  & c_1: &\; s_1 \succ s_2 \succ s_3,\\
    s_2: &\; c_1 \succ c_2,  & c_2: &\; s_3 \succ s_2 \succ s_1,\\
    s_3: &\; c_2 \succ c_1.
\end{align*}

Suppose the current matching is $Y = \{(s_1, c_2), (s_2, c_1)\}$ and $k = 1$.  The cutoff student at $c_1$ is $s_2$, since all students ranked weakly above $s_2$ are assigned to colleges they weakly prefer to $c_1$, while the next student below $s_2$, namely $s_3$, is unmatched which is worse than $c_1$.  
For $c_2$, no student satisfies the cutoff conditions: the highest-priority student $s_3$ is not assigned to a college weakly better than $c_2$, so the cutoff student can be regarded as an \emph{imaginary student} with cutoff score $|S| + 1$.  
Although $s_1$, who has lower priority than $s_3$, is assigned to $c_2$, the matching $Y$ remains ER-$1$ because only $s_3$ has justified envy toward $s_1$.
\end{example}

We next present a procedure for updating a college's cutoff score, described in Algorithm~\ref{alg:update_cutoff}. 
Given a matching $Y$, a college $c$, and its current cutoff $B_c$,  
the algorithm examines the students ranked strictly below the current cutoff in descending order of priority. 
For each such student $s$, it checks whether $s$ strictly prefers college $c$ to her current assignment $Y_s$. 
If the student does \emph{not} strictly prefer $c$, then the cutoff score is decreased by one and the algorithm proceeds to the next student. 
If the student \emph{does} strictly prefer $c$, then the search terminates and the algorithm returns the current cutoff score. 
Finally, if the scan reaches the end of the priority list without encountering any student who strictly prefers $c$, the cutoff is automatically updated to $1$.

\begin{algorithm}[tb]
\caption{Update Cutoff Score}
\label{alg:update_cutoff}
\SetAlgoLined
\KwIn{Matching $Y$, college $c$, current cutoff $B_c$}
\KwOut{Updated cutoff $B_c'$}

Set $B_c' \gets B_c$\;
\ForEach{student $s$ ranked below $B_c$ in descending order of $\succ_c$}{
    \If{$s$ weakly prefers $Y(s)$ to $c$}{
        $B_c' \leftarrow B_c' -1$\;
        \If{$B_c' == 1$}{
        \Return{$B_c'$}
        }
    }
    \Else
    {
    \Return $B_{c}'$\;
    }
}
\end{algorithm}

\subsection{Algorithm Description}

We are now ready to formally present the $k$-admissible cutoff algorithm, as shown in Algorithm~\ref{alg:k_cutoff}. 
At a high level, the algorithm iteratively examines colleges one at a time and traverses each college’s priority list to identify students whose admission would result in a feasible, $k$-admissible, and individually improving matching. 
Instead of scanning each priority list from the top, the algorithm leverages the cutoff score to guide the search for such $k$-admissible students.

We now provide a detailed description of Algorithm~\ref{alg:k_cutoff}. 
The procedure begins with an empty matching $Y$ and initializes the cutoff vector $B$ so that $B_c = |S| + 1$ for every college $c \in C$. 
It also maintains a \emph{college stack} $P$, which keeps track of colleges that cannot admit any additional students without violating feasibility or $k$-admissibility under the current matching $Y$.

At each step, the algorithm selects the college $c \in C \setminus P$ with the smallest index and examines students in $c$’s priority order, starting from the student ranked immediately below its current cutoff score $B_c$ and proceeding downward. 
A student $s$ is assigned to $c$ if the following conditions hold:
\begin{enumerate}
\item $s$ prefers $c$ to her current assignment $Y_s$; and
\item $s$ is $k$-admissible to $c$ with respect to $Y$, thereby preserving ER-$k$.
\end{enumerate}

If such a student is found, the matching is updated to $Y' = (Y \setminus Y_s) \cup \{(s, c)\}$, by reassigning $s$ from $Y_s$ to $c$, and the college stack $P$ is reset to empty.
Once all students in $c$’s priority list have been examined, the cutoff score of $c$ is updated according to Algorithm~\ref{alg:update_cutoff}, and $c$ is added to the stack $P$.

The algorithm terminates when $P$ contains all colleges, i.e., when every college has been examined and no further student can be admitted without violating feasibility or ER-$k$. 
At termination, the resulting matching $Y$ satisfies both ER-$k$ and NW-$k$.

\begin{algorithm}[tb]
\caption{\textsc{$k$-Admissible Cutoff Algorithm}}
\label{alg:k_cutoff}
\SetAlgoLined
\KwIn{Set of students $S$, set of colleges $C$, feasibility function $f$, integer $k$}
\KwOut{An ER-$k$ and NW-$k$ matching $Y$}

Initialize $Y \gets \emptyset$\;
Initialize cutoff student $B_c \gets |S| + 1$ for all $c \in C$\;
Initialize college stack $P \gets \emptyset$\;

\While{$P \neq C$}{
    Select $c \in C \setminus P$ with the smallest index\;
    
    \ForEach{student $s$ in $c$’s priority order starting from $B_c - 1$ downward}{
        \If{$s$ prefers $c$ to $Y_s$ \textbf{and} $s$ is $k$-admissible to $c$ w.r.t.\ $Y$}{
            Update $Y \gets (Y \setminus Y_s) \cup \{(s, c)\}$\;
            Reset $P \gets \emptyset$\;
        }
    }
    Update $B_c$ using Algorithm~\ref{alg:update_cutoff}\;
    $P \gets P \cup \{c\}$\;
}

\Return{$Y$}\;
\end{algorithm}

\begin{theorem}
    \label{theo:k-admissible}
    Algorithm~\ref{alg:k_cutoff}, the $k$-admissible cutoff algorithm, always computes a matching that satisfies both ER-$k$ and NW-$k$. 
    Moreover, the algorithm runs in polynomial time.
\end{theorem}

\begin{proof}
We first show that Algorithm~\ref{alg:k_cutoff} satisfies ER-$k$.
The algorithm starts from the empty matching, which trivially satisfies ER-$k$.
At each iteration, it considers a college $c \in C \setminus P$ and examines students in descending priority order below the current cutoff $B_c$.
A student $s$ is reassigned to $c$ only if $s$ is $k$-admissible to $c$ with respect to the current matching $Y$.
By definition of $k$-admissibility, such a reassignment preserves ER-$k$.
Since the matching is modified only through $k$-admissible moves, ER-$k$ is maintained throughout the execution of the algorithm.
Therefore, the matching returned by Algorithm~\ref{alg:k_cutoff} satisfies ER-$k$.

\medskip
We next show that the algorithm satisfies NW-$k$.
Suppose, for contradiction, that the matching $Y$ produced by Algorithm~\ref{alg:k_cutoff} does not satisfy NW-$k$.
Then there exists a student $s$ and a college $c$ with an empty seat such that $s$ can make a $k$-legitimate claim to $c$, i.e., assigning $s$ to $c$ preserves feasibility and ER-$k$.

Recall that the algorithm maintains a stack $P$ of colleges that have been fully examined under the current matching.
A college $c$ is added to $P$ only after the algorithm has verified that no student below the cutoff can be reassigned to $c$ without violating feasibility or ER-$k$.
Moreover, whenever the matching is updated, the stack $P$ is reset to empty, ensuring that all colleges are re-examined under the updated matching.

The algorithm terminates only when $P = C$, meaning that every college has been examined and found unable to admit any additional student via a $k$-admissible reassignment.
In particular, for the college $c$ involved in the alleged $k$-legitimate claim, the algorithm must have previously determined that no such reassignment is possible.
This contradicts the existence of a $k$-legitimate claim by $s$ to $c$.
Hence, the resulting matching satisfies NW-$k$.

\medskip
Finally, we analyze the running time of the algorithm.
Throughout the execution of Algorithm~\ref{alg:k_cutoff}, students are reassigned only to colleges they weakly prefer to their current assignments.
Since the length of each student’s preference list is no greater than the number of colleges, the number of times a student can be reassigned is bounded by the number of colleges.
Moreover, each iteration performs a polynomial amount of work in the size of the input.
Therefore, Algorithm~\ref{alg:k_cutoff} terminates in polynomial time.

This completes the proof of Theorem~\ref{theo:k-admissible}.
\end{proof}

\subsection{Running Example}
We now illustrate how to compute an ER-1 and NW-1 matching using Algorithm~\ref{alg:k_cutoff}.

\begin{example}
\label{example:k_cutoff}
Consider an instance with students $S = \{s_1, s_2, s_3, s_4\}$ and colleges $C = \{c_1, c_2, c_3\}$.  
The students' preferences and the colleges' priority orders are as follows:
\begin{align*}
    s_1: &\; c_3 \succ c_1 \succ c_2,  & c_1: &\; s_1 \succ s_2 \succ s_3 \succ s_4,\\
    s_2: &\; c_2 \succ c_3 \succ c_1,  & c_2: &\; s_1 \succ s_2 \succ s_3 \succ s_4,\\
    s_3: &\; c_3 \succ c_2 \succ c_1,  & c_3: &\; s_4 \succ s_3 \succ s_2 \succ s_1,\\
    s_4: &\; c_3 \succ c_2 \succ c_1.
\end{align*}
Each college has capacity of 2, and at most three students can be matched in total.

\paragraph{Initialization.}
The algorithm begins by setting the initial cutoff score for each college $c$ to $B_c = |S| + 1 = 5$.  
The initial matching is empty, $Y := \varnothing$, and the college stack is initialized as $P := \varnothing$.  
The stack $P$ records the colleges that have been examined and cannot admit any additional students without violating feasibility or 1-admissibility.

\paragraph{Step 1.}
The algorithm first selects college $c_1$, the college with the smallest index.  
College $c_1$ begins by examining the student corresponding to $B_{c_1} - 1 = 4$, namely $s_1$.  
Since $s_1$ is unmatched and 1-admissible to $c_1$, the algorithm assigns $s_1$ to $c_1$, yielding the matching
\[
Y_1 = \{(s_1, c_1)\}.
\]

College $c_1$ next examines $s_2$, who is also unmatched and 1-admissible.  
Admitting $s_2$ produces the matching
\[
Y_2 = \{(s_1, c_1), (s_2, c_1)\}.
\]

The following student in the priority list, $s_3$, cannot be admitted because $c_1$'s capacity would be exceeded.  
The same holds for $s_4$.  
After all students below the current cutoff have been examined, no further assignments to $c_1$ are possible in this step, so $c_1$ is added to the college stack: $P = \{c_1\}$.  
The cutoff score for $c_1$ is then updated to $3$, corresponding to student $s_2$, and the cutoff vector becomes $B = (3, 5, 5)$.

\paragraph{Step 2.}
The algorithm next selects college $c_2 \in C \setminus P$.  
College $c_2$ begins by examining $s_1$ (the student corresponding to $B_{c_2}-1 = 4$).  
Student $s_1$ prefers her current assignment at $c_1$ to $c_2$, and thus no reassignment occurs.  

College $c_2$ then examines $s_2$, who prefers $c_2$ over her current assignment $c_1$ and is 1-admissible.  
Reassigning $s_2$ to $c_2$ yields the matching
\[
Y_3 = \{(s_1, c_1), (s_2, c_2)\}.
\]
Since a reassignment occurred, the college stack is reset to $P = \varnothing$.

Continuing down the priority list, $c_2$ examines $s_3$, who is unmatched and 1-admissible, and therefore assigns her to $c_2$:
\[
Y_4 = \{(s_1, c_1), (s_2, c_2), (s_3, c_2)\}.
\]

The next student, $s_4$, cannot be admitted due to the total capacity constraint.  
Thus, after completing its scan, $c_2$ is added to the college stack: $P = \{c_2\}$.  
The cutoff score of $c_2$ is then updated to $2$, and the cutoff vector becomes $B = (3, 2, 5)$.

\paragraph{Step 3.} 
Next, the algorithm revisits $c_1 \in C \setminus P$.  
College $c_1$ examines $s_3$, who is now assigned to $c_2$ and does not strictly prefer $c_1$, so no reassignment occurs.  
The next student, $s_4$, also cannot be admitted due to the total capacity constraint.  
After completing its scan, college $c_1$ is added to the stack: $P = \{c_1, c_2\}$.  

Although no new assignments were made, the cutoff score of $c_1$ is updated.  
Since $s_3$ (the student immediately below the previous cutoff) is now assigned to a college she weakly prefers to $c_1$, the cutoff decreases from $3$ to $2$.  
Thus the cutoff vector becomes $B = (2, 2, 5)$.

\paragraph{Step 4.} 
The algorithm then considers $c_3 \in C \setminus P$.  
College $c_3$ begins by examining $s_4$ (corresponding to $B_{c_3}-1 = 4$), who cannot be admitted because the total matching capacity has already been reached.  

Next, $c_3$ examines $s_3$, who strictly prefers $c_3$ to her current assignment at $c_2$ and is 1-admissible.  
Reassigning $s_3$ to $c_3$ produces the matching:
\[
Y_5 = \{(s_1, c_1), (s_2, c_2), (s_3, c_3)\}.
\]
Since a reassignment occurred, the college stack is reset to $P = \varnothing$.

College $c_3$ continues by examining $s_2$, who prefers to remain at $c_2$, and then $s_1$, who strictly prefers $c_3$ and is 1-admissible.  
Assigning $s_1$ to $c_3$ yields:
\[
Y_6 = \{(s_1, c_3), (s_2, c_2), (s_3, c_3)\}.
\]

Because $s_4$ remains unmatched and strictly prefers $c_3$, the cutoff score at $c_3$ stays at $5$, and the cutoff vector remains $B = (2, 2, 5)$.

\paragraph{Step 5.} 
Since the college stack is empty, the algorithm restarts with $c_1$.  
College $c_1$ examines the students below its cutoff but cannot admit any of them due to feasibility and 1-admissibility constraints, so it is added to the stack: $P = \{c_1\}$.  
Next, the algorithm proceeds to $c_2$, which likewise cannot admit any additional student, and so $c_2$ is added: $P = \{c_1, c_2\}$.  
Finally, $c_3$ also cannot admit any further student, and the stack becomes $P = C$.  

\paragraph{Final Matching.} 
Because all colleges are now in the stack, the algorithm terminates.  
This indicates that no college can admit any additional student without violating feasibility or 1-admissibility.  
The final ER-1 and NW-1 matching produced by the algorithm is
\[
Y_6 = \{(s_1, c_3), (s_2, c_2), (s_3, c_3)\},
\]
and the final cutoff vector is $B = (2, 2, 5)$.
\end{example}

\section{$k$-Admissible College-proposing Deferred Acceptance}
In this section, we present an alternative perspective on the $k$-Admissible Cutoff Algorithm by interpreting it as a constrained variant of the college-proposing deferred acceptance (DA) algorithm. Under this interpretation, colleges iteratively make proposals to students, but unlike in the standard DA algorithm, proposals are restricted to students who satisfy the $k$-admissibility condition.

Informally, the DA-based interpretation of the $k$-Admissible Cutoff Algorithm proceeds as follows. In each iteration, a college $c$ that is able to make a new proposal selects a student $s$ who is $k$-admissible to $c$. The student $s$ then compares the proposal from $c$ with her current assignment and tentatively accepts the option she prefers. The matching $Y$ is updated accordingly, and the set of $k$-admissible students is recomputed for the affected colleges. The algorithm terminates when no college can make any further $k$-admissible proposals.

\subsection{Formal Description}

The $k$-Admissible College-Proposing Deferred Acceptance ($k$-SPDA) algorithm closely parallels the $k$-admissible cutoff algorithm and can be interpreted as a constrained variant of the college-proposing deferred acceptance (DA) algorithm. It maintains the same core data structure, the college stack $P$, but operates explicitly through proposal and acceptance steps.

Initially, the matching $Y$ is empty.
At each iteration, the algorithm selects the college $c \in C \setminus P$ with the smallest index and allows $c$ to traverse its priority list, starting from the student ranked immediately below its cutoff $B_c$ and proceeding downward.
For each student $s$ below the cutoff, college $c$ makes a proposal to $s$ if and only if $s$ is $k$-admissible to $c$ with respect to the current matching $Y$. Upon receiving a proposal, student $s$ compares the new offer $(s,c)$ with her current tentative assignment $Y_s$ and tentatively retains the option she prefers. If $s$ accepts the proposal from $c$, the matching $Y$ is updated accordingly and the college stack $P$ is reset to empty, allowing all colleges to resume proposing under the updated matching. If $s$ rejects the proposal, college $c$ proceeds to the next student in its priority order.
If college $c$ exhausts its priority list without any successful proposal, it is added to the stack $P$, indicating that $c$ cannot improve its assignment under the current matching via a $k$-admissible proposal. As in the $k$-admissible cutoff algorithm, the procedure terminates when $P = C$, that is, when no college can make any further $k$-admissible proposals.

\begin{algorithm}[tb]
\caption{\textsc{$k$-Admissible College-Proposing Deferred Acceptance ($k$-SPDA)}}
\label{alg:k_DA}
\SetAlgoLined
\KwIn{Set of students $S$, set of colleges $C$, feasibility function $f$, integer $k$}
\KwOut{A feasible matching $Y$}

\BlankLine
Initialize the matching $Y \gets \emptyset$\;
Initialize cutoff $B_c \gets |S| + 1$ for all $c \in C$\;
Initialize college stack $P \gets \emptyset$\;

\BlankLine
\While{$P \neq C$}{
    Select the college $c \in C \setminus P$ with the smallest index\;
    
    \ForEach{student $s$ in $c$’s priority order starting from $B_c - 1$ downward}{
        \If{$s$ is $k$-admissible to $c$ w.r.t.\ $Y$}{
            \tcp{College proposal}
            College $c$ proposes to student $s$\;
            
            \If{$(s,c) \succ_s Y_s$}{
                \tcp{Student selection}
                Update the matching:
                $Y \gets (Y \setminus Y_s) \cup \{(s,c)\}$\;
                Reset $P \gets \emptyset$\;
            }
        }
    }
    
    Update $B_c$ using Algorithm~\ref{alg:update_cutoff}\;
    $P \gets P \cup \{c\}$\;
}
\Return{$Y$}\;
\end{algorithm}

\section{Experiments}

In this section, we present simulation results for the $k$-admissible cutoff algorithm and investigate how the parameter $k$ affects matching outcomes in terms of fairness and non-wastefulness.

\subsection{Dataset Generation}

We generate synthetic instances of a college choice problem with regional caps, a classical matching model that can be represented using hereditary constraints. The data-generating process is designed to reflect commonly studied features of centralized admissions systems while allowing controlled variation in the correlation of preferences, priorities and capacity heterogeneity. We next describe the instance generation procedure in detail.

\paragraph{Student Preferences}
The preferences of students are modeled as complete rankings over all colleges. Preferences are generated independently according to a Mallows model with repeated insertion \citep{LuCr11a}. Let $\succ_S^{\mathrm{ref}}$ denote a common reference ranking over colleges. For each student $s \in S$, preferences are drawn as
\[
\succ_s \sim \mathcal{M}(\succ_S^{\mathrm{ref}}, \phi_S),
\]
where $\phi_S \in (0,1]$ controls the degree of dispersion around the reference ranking. Smaller values of $\phi_S$ generate more heterogeneous preferences, while larger values induce stronger correlation across students.

\paragraph{College Priorities}
Each college $c \in C$ is assigned a priority ordering over students, independently drawn from a Mallows model,
\[
\succ_c \sim \mathcal{M}(\succ_C^{\mathrm{ref}}, \phi_C),
\]
where $\succ_C^{\mathrm{ref}}$ is a reference ranking of students and $\phi_C \in (0,1]$ governs the degree of correlation in priorities across colleges.

\paragraph{College Capacities}
We implement three methods for generating college capacities.
In the \emph{fixed-capacity} method, all colleges are assigned an identical capacity equal to
\[
q_c = \max\left\{1, \left\lfloor \rho_C \cdot |S| / |C| \right\rfloor \right\},
\]
where $\rho_C \in (0,1]$ is the target supply--demand ratio.

In the \emph{uniform-capacity} method, each college capacity is drawn independently from a discrete uniform distribution over a prespecified integer range.

In our experiments, we use the \emph{normal-capacity} method. The total capacity is first set to
\[
C_{\mathrm{tot}} = \max\left\{ |C|, \left\lfloor \rho_C \cdot |S| \right\rfloor \right\}.
\]
Let $\mu = C_{\mathrm{tot}}/|C|$ denote the mean capacity per college, and let $\sigma = \mu \times \texttt{std\_ratio}$, where \texttt{std\_ratio} defaults to $0.1$ if unspecified. Initial college capacities are drawn independently from a normal distribution $\mathcal{N}(\mu,\sigma)$, truncated below at $1$, and rounded to integers. Since rounding may change the total capacity, capacities are subsequently adjusted by incrementing or decrementing randomly selected colleges (while maintaining a minimum capacity of one) until the total capacity equals $C_{\mathrm{tot}}$.

\paragraph{Regions and Regional Capacities}
The number of regions is determined as
\[
|R| = \max\left\{1, \left\lfloor \texttt{region\_ratio} \cdot |S| \right\rfloor \right\},
\]
where \texttt{region\_ratio} is set to $0.02$ in the experiments.

Colleges are assigned to regions in a deterministic round-robin manner. Specifically, colleges are indexed arbitrarily and assigned sequentially to regions $1,2,\ldots,|R|$, cycling back to region~$1$ after region~$|R|$ until all colleges are assigned. Let $C(r) \subseteq C$ denote the set of colleges assigned to region $r \in R$.

Regional constraints are imposed as aggregate upper bounds. For each region $r$, the regional capacity is defined as a fixed fraction of the total capacity of colleges in that region:
\[
Q_r = \left\lfloor \texttt{region\_capacity\_ratio} \cdot \sum_{c \in C(r)} q_c \right\rfloor,
\]
where \texttt{region\_capacity\_ratio} is set to $0.6$ in the experiments.

A matching is feasible if and only if it respects both individual college capacity constraints and the regional capacity constraint for every region.

\subsection{Measurement}

To evaluate the effect of varying $k$, we report outcome measures that capture both fairness (ER-$k$) and non-wastefulness (NW-$k$).
For each parameter configuration, we generate 50 independent instances and report average outcome measures, as described below.

We first report ER-related measures. \emph{Average envy received} is the average number of students who envy a given student, while \emph{maximum envy received} is the largest number of students who envy any single student. The latter corresponds to the smallest $k$ for which the matching satisfies ER-$k$. Smaller values of this measure therefore indicate stronger fairness guarantees.

Next, we report a worst-case measure of non-wastefulness.
\emph{Maximum objections} records the largest number of objections that would arise if a student were to claim an empty seat at some college.
Formally, it is defined as
\[
\max_{\substack{s \in S,\; c \in C:\\ s \text{ claims an empty seat at } c}}
\left|
\left\{
s' \in S \;\middle|\; s' \text{ objects to reassigning } s \text{ to } c
\right\}
\right|.
\]

When the matching satisfies NW-$(n-1)$, there exists no pair $(s,c)$ with an empty seat that can be used to compute this maximum, since no student can claim any empty seat without violating non-wastefulness. This value corresponds to the smallest integer $k$ for which the matching satisfies NW-$k$.
Specifically, for each empty seat, we hypothetically assign a student to the corresponding college and count the number of students who would object to this assignment.
In general, this number ranges from $0$ to $|S|-1$.
A larger admissible value of $k$ indicates that students have more opportunities to claim empty seats without generating excessive objections, and therefore reflects a stronger degree of non-wastefulness.

Finally, we report aggregate measures. \emph{Total envy} counts the total number of justified envy relations in the matching, while \emph{total claims} counts the total number of empty-seat claims that would generate objections. For both measures, smaller values indicate better overall performance.

\subsection{Experimental Results}

We report experimental results in Table~\ref{tab:experiment:200-10-ktradeoff} for markets with 200 students and 10 colleges.
The parameters $\phi_s$ and $\phi_c$ denote the Mallows dispersion parameters governing student preferences and college priorities, respectively, with larger values corresponding to more correlated rankings.
For each value of $k$, the table reports ER-related measures (average and maximum envy received), the guaranteed number of objections under the NW-$k$ notion, and aggregate outcome statistics (total envy and total claims).

The table illustrates a clear trade-off between fairness and non-wastefulness as $k$ varies.
As $k$ increases, the matching allows more envy but progressively eliminates empty-seat claims, reflected in a sharp rise in the number of guaranteed objections and a decline in total claims.

Focusing on $k=10$, the results demonstrate that intermediate values of $k$ achieve a favorable balance between the two objectives.
When $\phi_s = 0.7$ and $\phi_c = 0.5$, the average envy received is $5.04$, while the maximum envy received is $10$, indicating that the outcome satisfies ER-$10$.
At the same time, the guaranteed number of objections reaches $177.64$, close to its maximum of $199$, suggesting that only very few empty-seat claims can be sustained.
Aggregate measures are also favorable: total envy remains moderate at $14.42$, and total claims are reduced to $0.52$, indicating that almost all empty seats are effectively eliminated.


A similar pattern holds when preferences are more correlated on the student side.
For $\phi_s = 0.9$ and $\phi_c = 0.5$, at $k=10$ the average envy received increases to $8.02$ and the maximum envy received remains at $10$, again satisfying ER-$10$.
The guaranteed number of objections rises sharply to $173.68$, while total claims drop to $0.34$, confirming strong non-wastefulness.
The increase in total envy reflects the higher degree of preference correlation, rather than a qualitative deterioration of fairness.

Overall, the table shows that intermediate values of $k$, such as $k=10$, nearly eliminate empty-seat claims while keeping envy uniformly bounded.
In contrast, very small values of $k$ preserve strong fairness at the cost of substantial waste, whereas very large values of $k$ yield only marginal improvements in non-wastefulness while allowing rapidly increasing envy, especially when student preferences are highly correlated.

\begin{table}[tb]
\centering
\setlength{\tabcolsep}{3pt}
\caption{200 students and 10 colleges and 4 regions}
\label{tab:experiment:200-10-ktradeoff}
\begin{tabular}{c c c | c c | c | c c}
\toprule
$\phi_s$ & $\phi_c$ & $k$ 
& avg. envy & max envy 
& maximum
& total 
& total \\
 &  &  
& received & received
& objections 
& envy 
& claims \\
\midrule


\multirow{8}{*}{0.7} & \multirow{8}{*}{0.5}
& 0   & 0.00 & 0  & 13.32 & 0.00 & 5.12 \\
& & 1   & 0.56 & 1  & 34.22 & 1.20 & 4.22 \\
& & 2   & 1.38 & 2  & 46.84 & 3.06 & 4.24 \\
& & 5   & 3.02 & 5  & 131.04 & 8.66 & 2.04 \\
& & 10  & 5.04 & 10 & 177.64 & 14.42 & 0.52 \\
& & 20  & 6.64 & 19 & 200.00 & 18.82 & 0.00 \\
& & 50  & 7.04 & 21 & 200.00 & 18.12 & 0.00 \\
& & 199 & 6.46 & 16 & 200.00 & 17.82 & 0.00 \\

\midrule

\multirow{8}{*}{0.7} & \multirow{8}{*}{0.7}
& 0   & 0.00 & 0  & 9.98  & 0.00 & 5.38 \\
& & 1   & 0.58 & 1  & 30.44 & 1.26 & 5.38 \\
& & 2   & 1.32 & 2  & 51.30 & 2.60 & 3.70 \\
& & 5   & 3.38 & 5  & 97.34 & 8.72 & 1.56 \\
& & 10  & 5.68 & 10 & 162.42 & 16.78 & 0.72 \\
& & 20  & 8.28 & 19 & 200.00 & 23.46 & 0.00 \\
& & 50  & 7.84 & 21 & 200.00 & 22.06 & 0.00 \\
& & 199 & 8.24 & 23 & 200.00 & 21.64 & 0.00 \\

\midrule

\multirow{7}{*}{0.7} & \multirow{7}{*}{0.9}
& 0   & 0.00 & 0  & 2.12  & 0.00 & 8.08 \\
& & 1   & 0.58 & 1  & 3.54  & 0.82 & 7.54 \\
& & 2   & 1.30 & 2  & 5.34  & 2.66 & 6.74 \\
& & 5   & 3.82 & 5  & 16.90 & 8.52 & 5.10 \\
& & 10  & 7.38 & 10 & 44.68 & 25.32 & 2.58 \\
& & 50  & 21.58 & 43 & 194.10 & 65.10 & 0.04 \\
& & 199 & 23.18 & 57 & 200.00 & 73.84 & 0.00 \\

\midrule


\multirow{8}{*}{0.9} & \multirow{8}{*}{0.5}
& 0   & 0.00 & 0  & 1.38  & 0.00 & 11.18 \\
& & 1   & 0.72 & 1  & 2.48  & 1.60 & 10.74 \\
& & 2   & 1.64 & 2  & 3.32  & 4.90 & 8.42 \\
& & 5   & 4.54 & 5  & 14.12 & 22.78 & 4.26 \\
& & 10  & 8.02 & 10 & 173.68 & 49.04 & 0.34 \\
& & 20  & 9.20 & 15 & 200.00 & 54.44 & 0.00 \\
& & 50  & 9.06 & 15 & 200.00 & 57.04 & 0.00 \\
& & 199 & 8.82 & 16 & 200.00 & 56.80 & 0.00 \\

\midrule

\multirow{8}{*}{0.9} & \multirow{8}{*}{0.7}
& 0   & 0.00 & 0  & 1.56  & 0.00 & 12.04 \\
& & 1   & 0.68 & 1  & 2.64  & 1.60 & 10.38 \\
& & 2   & 1.66 & 2  & 3.40  & 5.14 & 9.18 \\
& & 5   & 4.50 & 5  & 10.58 & 22.74 & 4.56 \\
& & 10  & 8.32 & 10 & 106.62 & 48.76 & 0.86 \\
& & 20  & 10.48 & 18 & 189.26 & 60.68 & 0.06 \\
& & 50  & 10.86 & 23 & 200.00 & 64.42 & 0.00 \\
& & 199 & 11.34 & 20 & 200.00 & 62.08 & 0.00 \\

\midrule

\multirow{7}{*}{0.9} & \multirow{7}{*}{0.9}
& 0   & 0.00 & 0  & 1.70  & 0.00 & 16.74 \\
& & 1   & 0.62 & 1  & 2.32  & 1.38 & 15.78 \\
& & 2   & 1.66 & 2  & 3.32  & 5.02 & 14.68 \\
& & 5   & 4.70 & 5  & 6.66  & 21.90 & 10.32 \\
& & 10  & 9.10 & 10 & 13.04 & 56.18 & 5.88 \\
& & 50  & 32.74 & 50 & 177.76 & 161.46 & 0.16 \\
& & 199 & 36.90 & 80 & 200.00 & 168.82 & 0.00 \\

\bottomrule
\end{tabular}
\end{table}

\subsection{Analysis of the Experimental Results}

We analyze how the parameter $k$ affects the outcomes produced by the $k$-admissible cutoff algorithm, focusing on  fairness (ER-$k$) and non-wastefulness (NW-$k$). The results reveal systematic and robust patterns across market sizes and preference environments.

\paragraph{Effect of $k$ on Fairness.}
Across all experimental settings, both the average and the maximum ER-$k$ measures increase as $k$ grows. For small values of $k$, this increase is relatively mild, indicating that limited relaxations of fairness constraints introduce only modest envy among students. As $k$ becomes larger, ER-$k$ rises more sharply and eventually plateaus, suggesting that once a sufficient number of fairness constraints are relaxed, additional increases in $k$ do not generate qualitatively new sources of envy. This saturation effect is consistently observed across different values of $\phi_s$ and $\phi_c$.

\paragraph{Effect of $k$ on Non-Wastefulness.}
The non-wastefulness measure NW-$k$ exhibits a strongly nonlinear response to changes in $k$. When $k$ is small, NW-$k$ remains low, indicating that most claims on empty seats can still be blocked with relatively few objections. For intermediate values of $k$, NW-$k$ increases rapidly, reflecting growing inefficiencies in seat allocation and a reduced ability to resolve justified claims. For sufficiently large $k$, NW-$k$ saturates at an upper bound determined by market size, implying that non-wastefulness constraints become effectively inactive in this regime.

\paragraph{Fairness and Non-wastefulness Trade-off.}
The joint behavior of ER-$k$ and NW-$k$ reveals a clear trade-off induced by $k$. Small increases in $k$ yield limited efficiency gains while preserving relatively strong fairness guarantees. In contrast, intermediate values of $k$ correspond to a region in which fairness deteriorates quickly while non-wastefulness is only partially preserved. Once $k$ becomes large, both fairness and non-wastefulness are substantially weakened, placing outcomes in a region that is dominated from both perspectives. This indicates that the most meaningful trade-offs occur for moderate values of $k$.

\paragraph{Interaction with Preference Correlation.}
The sensitivity of both ER-$k$ and NW-$k$ to $k$ depends critically on the preference correlation parameters $\phi_s$ and $\phi_c$. Higher values of $\phi_s$ amplify the growth of justified envy as $k$ increases, reflecting the fact that correlated student preferences intensify competition for popular schools. Similarly, higher values of $\phi_c$ accelerate the deterioration of non-wastefulness, as correlated college preferences tighten capacity constraints. As a result, identical values of $k$ can lead to substantially different outcomes across preference environments.

\section{Conclusion}
In this paper, we study two-sided matching markets under general upper bounds where fairness and non-wastefulness are incompatible in general. We introduce symmetric parametric relaxations, ER-$k$ and NW-$k$, and show that for any fixed $k$ there exists a matching satisfying both notions under hereditary constraints. We also provide polynomial-time algorithms to compute such matchings.

Our work opens several directions for future research. First, while ER-$k$ constraints admit a natural formulation using integer programming, it remains an open question how to efficiently encode NW-$k$ constraints within an integer programming framework. Second, although our results establish existence, it is unclear whether a student-optimal matching exists among all matchings satisfying both ER-$k$ and NW-$k$, analogous to classical optimality results under stability. Third, while we introduce a $k$-admissible college-proposing deferred acceptance mechanism, it is an open problem whether a student-proposing deferred acceptance algorithm can be designed to compute an ER-$k$ and NW-$k$ matching. Finally, an important direction is to investigate incentive properties in this relaxed framework, including whether meaningful forms of strategy-proofness, such as strategy-proofness with respect to truncation, can be guaranteed.

\section*{Acknowledgments}
This work was partially supported by the Japan Science and Technology Agency under the ERATO Grant Number JPMJER2301, and by the Japan Society for the Promotion of Science under Grant Numbers JP21K17700 and JP25K21277. 

\bibliographystyle{plainnat}
\bibliography{reference}  

@article{ABB24a,
  title={Cutoff stability under distributional constraints with an application to summer internship matching},
  author={H. Aziz and A. Baychkov and P. Bir{\'o}},
  journal={Mathematical Programming},
  volume={203},
  number={1},
  pages={247--269},
  year={2024},
  publisher={Springer}
}

@inproceedings{ACGS18a,
	author = {H. Aziz and J. Chen and S. Gaspers and Z. Sun},
	booktitle = {International Conference on Autonomous Agents and Multiagent Systems},
	pages = {964--972},
	title = {Stability and Pareto Optimality in Refugee Allocation Matchings},
	year = {2018}}

@article{AbSo03b,
	author = {A. Abdulkadiro{\u{g}}lu and T. S{\"o}nmez},
	journal = {American Economic Review},
	number = {3},
	pages = {729--747},
	title = {School Choice: A Mechanism Design Approach},
	volume = {93},
	year = {2003}}

@article{AyBo21a,
	author = {O. Ayg{\"u}n and I. B{\'o}},
      volume={13},
      number={3},
      pages={1--28},
      year={2021},
	journal = {American Economic Journal: Microeconomics},
	title = {College admission with multidimensional privileges: The {B}razilian affirmative action case},
	}

@article{AyTu17a,
	author = {O. Ayg{\"u}n and B. Turhan},
	journal = {American Economic Review},
	number = {5},
	pages = {210--13},
	title = {{Large-Scale Affirmative Action in School Choice: Admissions to IITs in India}},
	volume = {107},
	year = {2017}}

@article{AzSu25a,
  author       = {H. Aziz and
                  Z. Sun},
  title        = {Multi-rank smart reserves: {A} general framework for selection and
                  matching diversity goals},
  journal      = {Artif. Intell.},
  volume       = {339},
  pages        = {104274},
  year         = {2025}
}

@article{BCC+19b,
  author       = {S. Baswana and
                  P. P. Chakrabarti and
                  S. Chandran and
                  Y. Kanoria and
                  U. Patange},
  title        = {Centralized Admissions for Engineering Colleges in India},
  journal      = {{INFORMS} J. Appl. Anal.},
  volume       = {49},
  number       = {5},
  pages        = {338--354},
  year         = {2019},
    url          = {https://doi.org/10.1287/inte.2019.1007},
  doi          = {10.1287/INTE.2019.1007}
}

@article{BFIM10a,
	author = {P. Bir{\'o} and T. Fleiner and R. W. Irving and D. F. Manlove},
	issn = {0304-3975},
	journal = {Theoretical Computer Science},
	number = {34},
	pages = {3136 - 3153},
	title = {The College Admissions problem with lower and common quotas},
	volume = {411},
	year = {2010}}

@article{CKM+22a,
  title={Impossibility of weakly stable and strategy-proof mechanism},
  author={S. Cho and M. Koshimura and P. Mandal and K. Yahiro and M. Yokoo},
  journal={Economics Letters},
  volume={217},
  pages={110675},
  year={2022},
  publisher={Elsevier}
}

@inproceedings{CKL+24a,
  author       = {S. Cho and
                  K. Kimura and
                  K. Liu and
                  K. Liu and
                  Z. Liu and
                  Z. Sun and
                  K. Yahiro and
                  M. Yokoo},
  title        = {Fairness and Efficiency Trade-off in Two-sided Matching},
  booktitle    = {Proceedings of the 23rd International Conference on Autonomous Agents
                  and Multiagent Systems, {AAMAS} 2024, Auckland, New Zealand, May 6-10,
                  2024},
  pages        = {372--380},
  publisher    = {International Foundation for Autonomous Agents and Multiagent Systems
                  / {ACM}},
  year         = {2024}
}

@article{DKT23a,
  title={Matching Mechanisms for Refugee Settlement},
  author={D. Delacr{\'e}taz and S. Kominers and A. Teytelboym},
  journal={American Economic Review},
  volume={113},
  number={10},
  pages={2689--2717},
  year={2023}
}

@article{EHYY14a,
	author = {L. Ehlers and I. E. Hafalir and M. B. Yenmez and M. A. Yildirim},
	journal = {Journal of Economic Theory},
	pages = {648--683},
	publisher = {Elsevier},
	title = {School choice with controlled choice constraints: Hard bounds versus soft bounds},
	volume = {153},
	year = {2014}}

@article{GaSh62a,
	author = {D. Gale and L. S. Shapley},
	journal = {The American Mathematical Monthly},
	number = {1},
	pages = {9--15},
	publisher = {Taylor \& Francis},
	title = {College admissions and the stability of marriage},
	volume = {69},
	year = {1962}}

@inproceedings{GKK+15a,
  author       = {M. Goto and
                  F. Kojima and
                  R. Kurata and
                  A. Tamura and
                  M. Yokoo},
  title        = {Designing Matching Mechanisms under General Distributional Constraints},
  booktitle    = {Proceedings of the Sixteenth {ACM} Conference on Economics and Computation,
                  {EC} '15, Portland, OR, USA, June 15-19, 2015},
  pages        = {259--260},
  publisher    = {{ACM}},
  year         = {2015}
}

@article{IHZ+19a,
	author = {A. Ismaili and N. Hamada and Y. Zhang and T. Suzuki and M. Yokoo},
	journal = {Journal of Artificial Intelligence Research},
	pages = {393--421},
	title = {Weighted Matching Markets with Budget Constraints},
	volume = {65},
	year = {2019}
}

@inproceedings{KaIw17a,
  author       = {Y. Kawase and
                  A. Iwasaki},
  title        = {Near-Feasible Stable Matchings with Budget Constraints},
  booktitle    = {Proceedings of the Twenty-Sixth International Joint Conference on
                  Artificial Intelligence, {IJCAI} 2017, Melbourne, Australia, August
                  19-25, 2017},
  pages        = {242--248},
  year         = {2017}
}

@inproceedings{KaIw18a,
  author       = {Y. Kawase and
                  A. Iwasaki},
  title        = {Approximately Stable Matchings With Budget Constraints},
  booktitle    = {Proceedings of the Thirty-Second {AAAI} Conference on Artificial Intelligence,
                  (AAAI-18), the 30th innovative Applications of Artificial Intelligence
                  (IAAI-18), and the 8th {AAAI} Symposium on Educational Advances in
                  Artificial Intelligence (EAAI-18), New Orleans, Louisiana, USA, February
                  2-7, 2018},
  pages        = {1113--1120},
  publisher    = {{AAAI} Press},
  year         = {2018}
}

@article{KTY18a,
  author       = {F. Kojima and
                  A. Tamura and
                  M. Yokoo},
  title        = {Designing matching mechanisms under constraints: An approach from
                  discrete convex analysis},
  journal      = {J. Econ. Theory},
  volume       = {176},
  pages        = {803--833},
  year         = {2018}
}

@article{KaKo15a,
	author = {Y. Kamada and F. Kojima},
	journal = {The American Economic Review},
	number = {1},
	pages = {67--99},
	publisher = {American Economic Association},
	title = {Efficient matching under distributional constraints: Theory and applications},
	volume = {105},
	year = {2015}}

@article{KaKo17b,
	author = {Y. Kamada and F. Kojima},
  title        = {Stability concepts in matching under distributional constraints},
  journal      = {Journal of Economic Theory},
  volume       = {168},
  pages        = {107--142},
  year         = {2017}
}

@article{KaKo23a,
    author = {Kamada, Y. and Kojima, F.},
    title = "{Fair Matching under Constraints: Theory and Applications}",
    journal = {The Review of Economic Studies},
    volume = {91},
    number = {2},
    pages = {1162-1199},
    year = {2023},
    month = {04},
}

@inproceedings{LuCr11a,
	author = {T. Lu and C. Boutilier},
	booktitle = {Proceedings of the 28th International Conference on Machine Learning (ICML-11)},
	pages = {145--152},
	title = {Learning Mallows models with pairwise preferences},
	year = {2011}}

@article{Roth84a,
	author = {A. E. Roth},
	journal = {Journal of Political Economy},
	number = {6},
	pages = {991--1016},
	publisher = {The University of Chicago Press},
	title = {The evolution of the labor market for medical interns and residents: a case study in game theory},
	volume = {92},
	year = {1984}}

@inproceedings{STM+23a,
  Author = {Z. Sun and Y. Takenami and D. Moriwaki and Y. Tomita and M. Yokoo},
  booktitle = {Proceedings of the 37th AAAI Conference on Artificial Intelligence, AAAI 2023},
  Title = {Daycare Matching in {J}apan: Transfers and Siblings},
  pages = {14487-14495},
  Year = {2023}}

@article{TKK25a,
  author       = {R. Takeshima and
                  K. Kimura and
                  A. Kuroki and
                  T. Wakasugi and
                  M. Yokoo},
  title        = {A New Relaxation of Fairness in Two-Sided Matching Respecting Acquaintance
                  Relationships},
  journal      = {CoRR},
  volume       = {abs/2508.15296},
  year         = {2025},
  url          = {https://doi.org/10.48550/arXiv.2508.15296},
  doi          = {10.48550/ARXIV.2508.15296},
  eprinttype    = {arXiv}
}

@STRING{aaai = { AAAI Conference}}

@STRING{aamas = { International Conference on Autonomous Agents and Multiagent Systems (AAMAS)}}

@STRING{aaai = { AAAI Conference on Artificial Intelligence (AAAI)}}

@STRING{aamas = { AAMAS Conference}}

@STRING{ijcai = { International Joint Conference on Artificial Intelligence (IJCAI)}}

@STRING{springer = {Springer-Verlag}}

@STRING{acm = {ACM Press}}

\end{document}